\documentclass{article}

\usepackage{Macros}
\usepackage{blindtext} 
\usepackage{bm}
\usepackage{float}
\usepackage{wrapfig}
\usepackage{amsthm,stmaryrd,amsmath}
\usepackage{color}

\usepackage{longtable}
\allowdisplaybreaks 

\newcommand{\revred}[1]{#1}

\newcommand{\statls}{\nonumber \\ & \qquad}	

\newcommand{\ls}{}	
\newcommand{\fls}[1]{#1}	
\newcommand{\hqquad}{}

\newcommand\numberthis{\addtocounter{equation}{1}\tag{\theequation}}

\newcommand{\sym}[2]{#1 & = \text{#2} \\}
\newcommand{\cm}{C_M}
\newcommand{\cp}{C_{\Phi}}
\begin{document}
\title{Non-asymptotic Error Bounds For Constant Stepsize Stochastic Approximation For Tracking Mobile Agents}
\author{Bhumesh Kumar, Vivek Borkar, Akhil Shetty
\thanks{VB is and BK, AS were with the Department of Electrical Engineering, IIT Bombay, Powai, Mumbai, Maharashtra 400076, India. BK is now with the Department of Electrical and Computer Engineering, University of Wisconsin at Madison, Madison, WI 53706, USA. AS is now with the Department of Electrical  Engineering and Computer Science, University of California at Berkeley, Cory Hall, Hearst Avenue, Berkeley, CA 94720, USA. Email: bkumar@wisc.edu, borkar.vs@gmail.com, shetty.akhil@berkeley.edu. Work of VB was supported in part by a J.\ C.\ Bose Fellowship,  CEFIPRA grant No. IFC/DST-Inria-2016-01/448 ``Machine Learning for Network Analytics''  and a grant for `\textit{Approximation for high dimensional optimization and control problems}' from the Department of Science and Technology, Government of India. }
}
%
\maketitle
\begin{abstract}
This work revisits the constant stepsize stochastic approximation algorithm for tracking a slowly moving target and obtains a bound for the tracking error that is valid for the entire time axis, using the Alekseev non-linear variation of constants formula. It is the first non-asymptptic bound for the entire time axis in the sense that it is not based on the vanishing stepsize limit and associated limit theorems unlike prior works, and captures clearly the dependence on problem parameters and the dimension.
\end{abstract}

\bigskip

\section{Introduction}

\subsection{Background}

Robbins and Monro proposed in \cite{robbins1951stochastic}  a stochastic iterative scheme
\begin{equation}
x_{n+1} = x_{n} + a_{n}\bs{h(x_n) + M_{n+1}}, \ n \geq 0, \label{SA}
\end{equation}
for finding the zero(s) of a function $h(\cdot)$ given its noisy evaluations, with $M_{n+1}$ being the measurement noise. By a clever choice of the stepsize sequence $\{a_n\}$, viz., those satisfying
\begin{equation}
\sum_{n}a_n = \infty, \ \sum_{n} a_n^2 < \infty, \label{step}
\end{equation}
they were able to show almost sure (a.s.) convergence of the scheme to a zero of $h$ under reasonable hypotheses. The scheme has since been a cornerstone of not only statistical computation, but also in a variety of engineering applications ranging from signal processing, adaptive control, to more recently, machine learning. See \cite{chen2006stochastic, borkar2008stochastic} for some recent pedagogical accounts of stochastic approximation. What makes it so popular is its typically low per iterate memory and computational requirement and ability to `average out' the noise, which makes it ideal for adaptive estimation/learning scenarios. 
A later viewpoint \cite{DerevitskiiFradkov}, \cite{Ljung} views (\ref{SA}) as a noisy discretization of the ordinary differential equation (ODE for short)
\begin{equation}
\dot{x}(t) = h(x(t)) \label{ode}
\end{equation}
with decreasing stepsize and argues that the errors due to discretization and noise are asymptotically negligible under (\ref{step}), so that it has the same asymptotic behaviour as (\ref{ode}). See \cite{borkar2008stochastic,borkarmeyn2000ode} for a fuller development of this approach.\\

The clean theory under (\ref{step}) notwithstanding, there has also been an interest and necessity to consider constant stepsize $a_n \equiv a > 0$. The strong convergence claims under (\ref{step}) can no longer be expected\footnote{barring some very special cases, e.g., when the right hand side of (\ref{SA}) is contractive uniformly w.r.t.\ the noise variable.}, e.g., for the simple case of $\{M_n\}$ being i.i.d.\ zero mean, the best one can hope for is convergence to a stationary distribution. What one can still expect is a high probability concentration around the desired target, viz., zero(s) of $h$, if the stepsize $a$ is small \cite{kushner1981asymptotic, kushner1997stochastic}. This is acceptable and in fact unavoidable in the important application area of tracking a slowly moving target or measuring a slowly exciting signal \cite{guo1995performance, benveniste1987design}, and other instances of learning in a slowly varying environment. This is because with decreasing stepsize, the algorithmic time scale, dictated by the decreasing stepsize, eventually becomes slower than the timescale on which the target is moving and thereby loses its tracking ability. The alternative of either frequent resets or adaptive loop gain is often not desirable because of the additional logic it requires, particularly when the algorithm is hard-wired \cite{utkin2017sliding,borkar2008stochastic}, and one settles for a judiciously chosen constant stepsize. Such schemes are a part of traditional signal processing and neural network algorithms \cite{eweda1994comparison, diamantaras1996principal, finnoff1993diffusion,  sharma1996asymptotic, kuan1991convergence, ng1999fast} and often show up in important applications such as quasi-stationary experimentation for meteorology \cite{chappell1986quasi}, slowly exciting physical wave measurement \cite{burke1971gravitational}, and more recently in online learning and non-stationary optimization \cite{zhu2017tracking,  wilson2018adaptive}. \revred{However, the focus in online learning is cumulative regret bounds instead of all time bounds}. \\

These developments have motivated analysis of constant stepsize schemes \cite{Bucklew, kushner1981averaging, kushner1981asymptotic, pflug1986stochastic, pflug1990nonasymptotic, borkarmeyn2000ode,   joslin2000law,  sharma1996asymptotic}  in the form of various limit theorems, (non) asymptotic analysis, law of iterated logarithm etc., but a convenient bound valid for all time, a useful metric for tracking applications in a slowly varying environment,
seems to be a topic of relatively recent interest \cite{zhu2017tracking}. Our objective here is to provide precisely one such bound.

\bigskip

\subsection{Comparison with Prior Art}

 As already mentioned, one of the main motivation of constant stepsize stochastic approximation has been their ability to track slowly moving environments. Not surprisingly, much of the early work has come from signal processing and control, most notably in adaptive filtering, and this continues to be its primary application domain. Some representative works are \cite{benveniste1987design}, \cite{farden1981tracking}, \cite{Guo1}, \cite{guo1995performance}, \cite{Guo2}, \cite{sharma1996asymptotic}, \cite{zhu2017tracking}, etc.\\

  Much of this work concerns tracking in specific models and the proposed schemes  usually have a very specific structure, e.g., linear. From purely theoretical angle, analyses appear in \cite{Bucklew}, \cite{joslin2000law}, \cite{kushner1981asymptotic}, \cite{pflug1986stochastic}, \cite{pflug1990nonasymptotic} among others. The emphasis of the latter is towards analyzing convergence properties in the small stepsize limit and the associated functional central limit theorem for fluctuations around the deterministic o.d.e.\ limit, except in case of \cite{joslin2000law}, which establishes a law of iterated logarithms, and \cite{pflug1990nonasymptotic}, which obtains confidence bounds for a specific choice of adaptive stepsizes and stopping rule. The latter is a non-asymptotic result as the title suggests, but in a different sense than us. \\

   In the context of tracking, the functional central limit theorem characterizing a Gauss-Markov process as a limit in law of suitably scaled fluctuations is also used for suggesting performance metrics for tracking application, see, e.g., \cite{benveniste1982measure}.\\

    More recently constant stepsize stochastic gradient and its variants have elicited interest in machine learning literature due to the possibility of using them in conjunction with iterate averaging to get better speed than decreasing stepsizes, see, e.g, \cite{BachMoulines}. The pros and cons of these have been discussed, e.g., in \cite{Laxmi}. This motivation, however, is not relevant for tracking because iterate averaging is also a stochastic approximation with decreasing stepsizes ($a_n = 1/(n+1)$ to be precise) and decreasing stepsizes is simply not an option here because the iterates will eventually become slower than the slowly varying signal and lose their tracking ability.\\

     Another strand of work analyzes tracking in the specific context of tracking the solution of an optimization problem when its parameters drift slowly \cite{wilson2018adaptive}.
Tracking problems have also been studied in the literature as regime switching stochastic approximations when the evolution is modulated by a Markov chain on a time scale equal to or faster than that of the algorithm. This situation has been analysed through mean squared error bounds \cite{yin2004regimeswitch,yin2009regimeswitch} and is close in spirit to ours.

\bigskip

\subsection{Our contributions}

Our main result is Theorem \ref{thm: mainthm}. The highlights of this result are as follows.\\

\begin{enumerate}
\item Our set-up is applicable to a very general scenario that includes unbounded correlated noise without any explicit evolution model, no explicit strong convexity or linearity assumptions regarding the dynamics being tracked, and so on, rendering it a more general framework than in prior work.\\

\item  We provide a bound valid for the entire time axis, not only for a finite time interval as in, e.g., `sample complexity' bounds, or purely asymptotic as in, e.g., cumulative regret bounds or asymptotic error bounds. That is, it holds uniformly for all $n, \ 0 \leq n < \infty$, not only for $n \leq$ some $N$ or in the $n \to \infty$ limit, which need not reflect the finite time behavior.
This is particularly relevant here because we are considering the problem of continuously  tracking a \textit{time-varying,  in particular, non-stationary} target. Furthermore, this is achieved under a very general noise model, viz., martingale differences which allow dependence across times, requiring only uncorrelatedness. Their conditional distributions given the past are required to satisfy an exponential moment bound that is satisfied by most standard distributions such as exponential, gaussian, their mixtures, etc., except the heavy tailed ones.\\

\item This bound is \textit{non-asymptotic}, i.e., it is
 derived for the actual constant stepsize $a > 0$ and not from an idealized limiting scenario based on a limit theorem for fluctuations in the $a\downarrow 0$ limit, as is often the case in prior studies.
To the best of our knowledge, ours is the first result to achieve this. Also, our derivation of the bound allows us to keep track of its dependence on problem parameters, dimension, etc.\ if needed.\\

\item We bound the \textit{exact} error which is given by the Alekseev formula, there is no approximation at this stage.  Furthermore, we analyze this error keeping the slow movement of the target being tracked in tact, without treating it as essentially static as, e.g., in \cite{benveniste1982measure}.

\end{enumerate}

As for potential avenues for improvement, we have the following observations:

\begin{enumerate}

\item It appears unlikely that the bounds that use Lipschitz constants etc., can be improved much, if at all. The moment bounds on martingale differences use state of the art martingale concentration inequalities and could improve if better inequalities become available. It may be noted that we assume exponential tails for the distributions of martingale differences. Stronger inequalities such as McDiarmid's inequality may be used under stronger hypotheses such as uniformly bounded martingale differences, see, e.g., \cite{BorkarBound}. On a different note, if we allow heavy tailed noise, one would get weaker claims using the corresponding, naturally weaker, concentration inequalities. Rather limited results are available here, see e.g., \cite{Joulin} and its application to stochastic approximation in \cite{Ananth}.\\

\item  One potential spot for improvement is in the use of the assumption $(\dagger)$ below, which entails a stability condition for a linearized dynamics which is time-dependent. Such conditions are available only under constraints on the time scale separation between fast dynamics  (of the algorithm) and the slow one (of the target). This, as argued later, is unavoidable, because there will be no tracking  otherwise. This fact necessitates such a condition or something close to it. The one we have used, due to Solo \cite{solo1994stability}, is the most general available to our knowledge. (Another class of sufficient conditions available is based on existence of Liapunov functions and not explicit like Solo's.)
\end{enumerate}

\bigskip

\subsection{Organization}

 We begin by describing the problem formulation in the next section. This is followed by the Alekseev formula as a non-linear generalization of the variation of constants formula, and  a key exponential stability assumption. A useful set of sufficient conditions for this assumption are recalled. Section 3 details the error analysis  characterizing the tracking behaviour,  developed through a sequence of lemmas and leading to the main result in section 4.  Section 5 concludes with some discussion. An appendix recalls a martingale concentration inequality used in the main text.

\bigskip

\subsection{Symbols and Notation}

The section number where the notation first appears is given in parentheses.
\begin{align*}
\sym{x_n}{Iterate at time $n$ (2.1)}
\sym{h(\cdot,\cdot)}{driving vector field of the tracking scheme (2.1)}
\sym{a}{Step-size (2.1)}
\sym{M_{n+1}}{Martingale difference noise (2.1)}
\sym{\varepsilon_{n+1}}{Additive error (2.1)}
\sym{\varepsilon^*}{Bound on $\|\varepsilon_{n+1}\|$ (2.1)}
\sym{y(\cdot)}{Slowly varying signal to be tracked (2.1)}
\sym{\epsilon}{Small ($\ll 1$) number controlling rate of $y(\cdot)$ (2.1)}
\sym{\gamma(\cdot)}{Vector field driving $y(\cdot)$ (2.1)}
\sym{C^*}{$\max\bbc{\sup_n\E\bs{\|x_n\|^2}^{1/2}, \sup_n\E\bs{\|x_n\|^4}^{1/4}}$ (2.1)}
\sym{C_{\gamma}}{$\sup_{t\geq 0}\|y(t)\|$ (2.1)}
\sym{\cm, \delta}{Constants featuring in the bound for $\nom{M_{n+1}}$ (2.1)}
\sym{\Phi}{Transition matrix of linear system (2.2)}
\sym{z(\cdot)}{Slowly varying equilibrium for the algorithm (2.3)}
\sym{d}{Dimension of $x_n$ and $z(\cdot)$ (2.1/2.3)}
\sym{\nabla}{Gradient operator (2.3)}
\sym{\cp,\beta}{Constants featuring in the exponential bound  for $\Phi$ (2.3)}
\sym{L_f}{Lipschitz constant for Lipschitz function $f$ (generic)}
\sym{G_f}{constant of linear growth for function $f$, i.e., $|f(x)| \leq C_f(1 + \|x\|)$. (generic)}
\sym{B_f}{$\sup_x|f(x)|$ for a bounded function $f$ (generic)}
\sym{K_{\gamma}}{$\max_{\{y\}\leq C_{\gamma}}\|\gamma(y)\|$ (3.2)}
\sym{O(\cdot)}{Big O notation (3.3)}
\sym{\mu}{$1/ \beta$ (3.3)}
\sym{K_1}{$L_{\tilde{h}}(1+C_{h}+C_{\gamma})+\varepsilon^*$ (3.3)}
\sym{K_2}{$\cp L_{\tilde{h}}$ (3.3)}
\sym{K_3}{$K_1 + L_{\gamma}a\ep$ (3.3)}
\sym{K_4}{$\max\bc{2\cm/\delta^2, \cm^2 / \delta^2}$ (3.3)}
\sym{K_5}{$K\cp^3L_D/\beta$ (3.4)}
\sym{K_6}{$\cp^3L_{\gamma}L_D\ep$ (3.4)}
\sym{K_7}{$\max\bc{24\cm \delta^4, 4\cm^2 \delta^4}$ (3.4)}
\sym{K_8}{$2\sqrt{6\cm}C_h^2/ \delta^2$ (3.4)}
\sym{K_{9}}{$\dfrac{\cm\gamma_1d^{1.5}}{\delta}$ (3.5)}
\end{align*}

\section{Preliminaries}
\subsection{The tracking problem}
We consider a constant step size stochastic approximation algorithm given by the $d$-dimensional iteration
\begin{equation}
x_{n+1}  = x_n + a \bs{h(x_n,y_n) + M_{n+1} + \varepsilon_{n+1}}, \ n \geq 0, \label{alg}
\end{equation}
for tracking a slowly varying signal governed by
\begin{equation}
\dot{y}(t) = a\ep\gamma(y(t)), \label{slow}
\end{equation}
with \revred{$0 < a < 1$}, $0 < \ep \ll 1$. Also,  $y_n := y(n), n \geq 0,$ the trajectory of (\ref{slow}) sampled at unit\footnote{without loss of generality} time intervals coincident with the clock of the  above iteration, with slight abuse of notation. We assume that $y(t), t \geq 0,$ remains in a bounded set. The term $\varepsilon_{n+1}$ represents an added bounded component attributed to  possible numerical errors (e.g., error in gradient estimation in case of stochastic gradient algorithms \cite{borkar2008stochastic}). We assume the following:
\begin{itemize}
\item The smallness condition on $\ep$ ensures a separation of time scale between the two evolutions (\ref{alg}) and (\ref{slow}), in particular (\ref{slow}) has to be `sufficiently slow' in a sense to be made precise later.

\item $h: (x,y) \mapsto h(x,y)$ is twice continuously differentiable in $x$ with the first and second partial derivatives in $x$ bounded uniformly in $y$ in a compact set, and Lipschitz in $y$. A common example is where $h(x, y) = -(x - y)$ corresponding to least mean square criterion for tracking in the above context with $x$, resp.\ $y$ standing for the states of the tracking scheme and the target resp.

\item $\gamma(\cdot)$ is Lipschitz continuous,

\item $C^* := \max\bbc{\sup_n\E\bs{\|x_n\|^2}^{1/2}, \sup_n\E\bs{\|x_n\|^4}^{1/4}} < \infty$. (See \cite{borkarmeyn2000ode} for sufficient conditions for uniform boundedness of second moments. Analogous conditions can be given for fourth moments.)

\item $C_{\gamma} := \sup_{t\geq 0}\|y(t)\| < \infty$,

\item there exists a constant $\varepsilon^* > 0$ such that
\begin{align}
\nom{\varepsilon_{n+1}} \leq \varepsilon^*, ~\forall n \geq 0, \label{eq: epbound}
\end{align}
\item   ${M_n}$  is a martingale difference sequence w.r.t.\ the increasing $\sigma$-fields
$$\F_n := \sigma(x_m, M_m, \varepsilon_m, m \leq n), n \geq 0,$$
and satisfies: there exist continuous functions $c_1, c_2 : \mathcal{R}^d \to (0, \infty)$ with $c_2$ being bounded away from $0$, such that
\begin{equation}
P(||M_{n+1}|| > u | \mathcal{F}_n) \leq c_1(x_n)e^{-c_2(x_n)u}, \ n \geq 0, \label{eqn: exp-bd}
\end{equation}
for all $u \geq v$ for a fixed, sufficiently large $v > 0$ (i.e., a sub-exponential tail) with
\begin{equation}
\sup_nE\left[c_1(x_n)\right] < \infty. \label{bddconst}
\end{equation}
\end{itemize}
In particular, (\ref{eqn: exp-bd}), (\ref{bddconst}) together imply that there exist $\delta, \cm > 0$ such that
\begin{equation}
E\left[e^{\delta\|M_{n+1}\|}\right] \leq \cm, \ n \geq 0.
\end{equation}
Using the Taylor expansion of the exponential function, we get
\begin{equation}
\sum_{m=0}^{\infty} \frac{\delta^m E\left[\|M_{n+1}\|^m \right]}{m!} \leq \cm, \ n \geq 0.
\end{equation}
As each term in the above summation is positive, we can conclude that for all $n, m \geq 0$,
\begin{equation}
 E[ ||M_{n+1}||^m ] \leq \dfrac{\cm m!}{\delta^m}. \label{eq:mgsingleMoment}
\end{equation}
We shall be interested in $m = 2, 4$. \\

These bounds will play an important role in our error analysis.
We next state a formula due to Alekseev  \cite{alekseev1961estimate} that captures the difference between the trajectory of a system and its (regular) perturbation, and may be viewed as a `non-linear variation of constants' formula.
\subsection{Alekseev's formula}

Consider the ODE
\begin{displaymath}
\dot{w}(t) = f(t,w(t)),~ t\geq 0,
\end{displaymath}
and its perturbed version,
\begin{displaymath}
\dot{u}(t) = f(t,u(t)) + g(t,u(t)),~ t\geq 0,
\end{displaymath}
where $f,g$ : $\mathcal{R} \times \mathcal{R}^d \mapsto \mathcal{R}^d $, with:
\begin{itemize}

\item $f(t, x)$ is measurable in $t$ and continuously differentiable in $x$ with bounded derivatives uniformly w.r.t.\ $t$, and,

\item  $g(t, x)$ is measurable in $t$ and Lipschitz in $x$ uniformly w.r.t.\ $t$.

\end{itemize}
Let $w(t,t_0,u_0)$ and $u(t,t_0,u_0)$ denote respectively the solutions to the above non-linear systems for $t \geq t_0$, satisfying  $w(t_0,t_0,u_0) = u(t_0,t_0,u_0) = u_0$. Then for $t \geq t_0$,
\begin{align}
u(t,t_0,u_0) &= w(t,t_0,u_0) \ls +
\int_{t_0}^{t}\Phi(t,s,u(s,t_0,u_0))g(s,u(s,t_0,u_0))\ds, \label{eq:alekseev}
\end{align}
where $\Phi(t,s,w_0)$ for any $w_0 \in \mathcal{R}^d$ is the fundamental matrix of the linearized system
\begin{align}
\dot{\phi}(t) = \frac{\partial f}{\partial w}(t,w(t,s,w_0))\phi(t),~t\geq s,
\end{align}
with $\Phi(s,s,w_0) = \mathcal{I}_d$, the d-dimensional identity matrix. That is, it is the unique solution to the matrix linear differential equation
$$\dot{\Phi}(t,s,w_0) = \frac{\partial f}{\partial w}(t,w(t,s,w_0))\Phi(t,s,w_0)$$
with the aforementioned initial condition at $t = s$. The equation (\ref{eq:alekseev}) is the Alekseev nonlinear variation of constants formula \cite{alekseev1961estimate} (see also Lemma 3, \cite{brauer1966perturbations}). \\

\revred{The generalization of Alekseev's nonlinear variation of constants for differing initial conditions \cite{borkar2018orbit} is given by
\begin{align}
u(t,t_0,u_0) &= w(t,t_0,w_0) + \Phi(t,t_0,u_0)(u_0 - w_0) \statls +
\int_{t_0}^{t}\Phi(t,s,u(s,t_0,u_0))g(s,u(s,t_0,u_0))\ds \label{eq:genalekseev}
\end{align}}
where the additional additive term captures the contribution due to differing initial conditions. This term will decay exponentially under our assumption $(\dagger)$  below.
\subsection{Perturbation analysis}

In view of the ODE approach described earlier, we consider  the candidate ODE
\begin{equation}
\dot{x}(t) = h(x(t),y) \label{param}
\end{equation}
where we have treated the $y$ component as frozen at a fixed value in view of its slow evolution (recall that $\ep << 1$). We assume that this ODE has a globally stable equilibrium $\lambda(y)$ where $\lambda$ is twice continuously differentiable with bounded first and second derivatives. (Typically, this can be verified by using the implicit function theorem.) In particular,
\[
h(\lambda(y),y) = 0~ \forall ~y \implies h(\lambda(y(t)), y(t)) = 0~ \forall~ t \geq 0.
\]
Define $z(t) = \lambda(y(t)), t \geq 0$. Then
\begin{align*}
\dot{z}(t) & = \ep a\nabla\lambda(y(t))\gamma(y(t))\\
& = a h(\lambda(y(t)), y(t)) + \ep a\nabla \lambda(y(t))\gamma(y(t))\\
& = a h(z(t), y(t)) + \ep a\nabla \lambda(y(t))\gamma(y(t))= a \tilde{h}(z(t),y(t))
\end{align*}
for
$$\tilde{h}(z,y) := h(z,y) + \ep \nabla\lambda(y)\gamma(y).$$
The corresponding Euler scheme would be
\begin{align*}
z_{n+1} = z_n + a \tilde{h}(z_n,y_n).
\end{align*}
The tracking algorithm (\ref{alg})  can therefore be equivalently written as:
\begin{align}
x_{n+1} & = x_n + a \bs{h(x_n,y_n) + M_{n+1} + \varepsilon_{n+1}}\\
& = x_n + a\bs{\tilde{h}(x_n,y_n) - \ep \nabla\lambda(y_n)\gamma(y_n)\ls  + M_{n+1} + \varepsilon_{n+1}}\\
& = x_n + a \bs{\tilde{h}(x_n,y_n) + \kappa_n(y_n)}, \label{eqn: x_SA}
\end{align}
where,
\begin{align}
\kappa_n(y_n) & = - \ep \nabla\lambda(y_n)\gamma(y_n) + M_{n+1} + \varepsilon_{n+1}.
\end{align}
Let $\x{}$ be the linearly interpolated trajectory of the stochastic approximation iterates such  that $\x{k} = x_k$. That is, for $t_n \equiv na \ \forall n$,
\begin{align}
\x{} = \x{n} + \frac{t-t_n}{a}\bs{\x{n+1}-\x{n}}, \ t \in [t_n, t_{n+1}]. \label{eqn: inter}
\end{align}
Then from (\ref{eqn: x_SA}), we get
\begin{align}
\x{n+1} & = \x{0} + \sum_{k=0}^{n}a\tilde{h}(\x{k},\y{k}) \ls - \sum_{k=0}^{n}a\ep \nabla\lambda(\y{k})\gamma(\y{k})
\statls + \sum_{k=0}^{n}aM_{k+1} + \sum_{k=0}^{n}a\varepsilon_{k+1}\\
& = \x{0} + \sum_{k=0}^{n}\int_{t_k}^{t_{k+1}} \tilde{h}(\x{k},\y{k})\ds \ls - \sum_{k=0}^{n}\int_{t_k}^{t_{k+1}}\ep \nabla\lambda(\y{k})\gamma(\y{k})\ds \statls +\sum_{k=0}^{n}\int_{t_k}^{t_{k+1}}M_{k+1}\ds + \sum_{k=0}^{n}\int_{t_k}^{t_{k+1}}\varepsilon_{k+1}\ds.
\end{align}
For $k \geq 0$ and $s \in [t_k,t_{k+1}]$, define perturbation terms:
\begin{align*}
\ze_1(s) & := \tilde{h}(\x{k},\y{k}) - \tilde{h}(\bar{x}(s),y(s)),\\
\ze_2(s) & := M_{k+1},\\
\ze_3(s) & := \varepsilon_{k+1},\\
\ze_4(s) & := -\ep \nabla\lambda(\y{k})\gamma(\y{k}).
\end{align*}
Thus
\begin{align*}
\x{n+1} & = \x{0} + \int_{t_0}^{t_{n+1}} \tilde{h}(\bar{x}(s),y(s))\ds \ls \\
& + \int_{t_0}^{t_{n+1}} \bbc{ \ze_1(s) + \ze_2(s) + \ze_3(s) + \ze_4(s)}\ds.
\end{align*}
Using (\ref{eqn: inter}),
\begin{align}
\x{} & = \x{0} + \int_{t_0}^{t} \tilde{h}(\bar{x}(s),y(s))\ds \ls +\int_{t_0}^{t} \bbc{\ze_1(s) + \ze_2(s) + \ze_3(s) + \ze_4(s)}\ds.
\end{align}
Define $$\Xi(t) = \ze_1(t) + \ze_2(t) + \ze_3(t) + \ze_4(t).$$
Consider the coupled systems
\begin{align}
\dot{z}(t) & = \tilde{h}(z(t),y(t)), \label{eqn: up1} \\
\dot{y}(t) & = \ep a \gamma(y(t)), \label{eqn: up2}
\intertext{and}
\dot{\bar{x}}(t) & = \tilde{h}(\bar{x}(t),y(t)) + \Xi(t) \label{eqn: p1},\\
\dot{y}(t) & = \ep a \gamma(y(t)). \label{eqn: p2}
\end{align}
The ODE (\ref{eqn: p1}) can be seen as a perturbation of the (\ref{eqn: up1}), with the perturbation term being $\Xi(t)$.
\newline
\newline
Let $D(\cdot,\cdot) \in \mR^{d \times d}$ denote the Jacobian matrix of $h$ (and therefore of $\tilde{h}$) in the first argument, and $\Gamma(\cdot) \in \mR^{d \times d}$ the Jacobian matrix of $\lambda$. Then the linearization or `equation of variation' of  (\ref{eqn: up1}) is
\begin{align}
\label{eqn: lin}
\dot{r}(t) = D(z(t), y(t))r(t).
\end{align}
For $t \geq s \geq 0$ and $x,y \in \mR^{d},$ let $\Phi(t,s;x_0,y_0)$ denote the fundamental matrix for the time varying linear system (\ref{eqn: lin}), i.e., the solution to the matrix-valued differential equation
\begin{align}
\label{eqn: matdiff}
\dot{\Phi}(t,s;x_0,y_0) = D(z(t),y(t))\Phi(t,s;x_0,y_0), ~t\geq s,
\end{align}
with initial condition $\Phi(s, s; x_0,y_0) = I$.
Then by Alekseev's formula,
\begin{align}
\bar{x}(t) & = z(t) + \revred{\Phi(t,t_0;\bar{x}(t_0),y_0)(\bar{x}(t_0)-z(t_0))} \ls + \int_{t_0}^{t}\Phi(t,s;\bar{x}(s),y(s))\Xi(s)ds. \nonumber
\end{align}
Define
\begin{align}
\revred{\varrho_n} & = \revred{\Phi(t_n,t_0;\bar{x}(t_0),y_0)(\bar{x}(t_0)-z(t_0))}\\
A_n & = \sum_{k=0}^{n-1}\int_{t_k}^{t_{k+1}}\Phi(t_n,s;\bar{x}(s),y(s)) \ls \bs{\tilde{h}(\x{k},\y{k}) - \tilde{h}(\bar{x}(s),y(s))}\ds, \label{eqn: andef}\\
B_n & = \sum_{k=0}^{n-1}\int_{t_k}^{t_{k+1}}\Phi(t_n,s;\bar{x}(s),y(s))M_{k+1}\ds, \label{eqn: bndef}\\
C_n & = \sum_{k=0}^{n-1}\int_{t_k}^{t_{k+1}}\Phi(t_n,s;\x{k},\y{k})M_{k+1}\ds, \label{eqn: cndef}\\
D_n & = \sum_{k=0}^{n-1}\int_{t_k}^{t_{k+1}}\Phi(t_n,s;\bar{x}(s),y(s))\varepsilon_{k+1}\ds, \label{eqn: dndef}\\
E_n & = \sum_{k=0}^{n-1}\int_{t_k}^{t_{k+1}}\Phi(t_n,s;\bar{x}(s),y(s))\ls \hqquad \hqquad \hqquad \times \ep \nabla\lambda(\y{k})\gamma(\y{k})\ds \label{eqn: endef}.
\end{align}
Then
\begin{align}
\x{n} & = {z}(t_n) + \int_{t_0}^{t_n}\Phi(t_n,s;\bar{x}(s),y(s))\ze_1(s)\ds \ls + \int_{t_0}^{t_n}\Phi(t_n,s;\bar{x}(s),y(s))\ze_2(s)\ds \nonumber\\ &+ \int_{t_0}^{t_n}\Phi(t_n,s;\bar{x}(s),y(s))\ze_3(s)\ds \ls + \int_{t_0}^{t_n}\Phi(t_n,s;\bar{x}(s),y(s))\ze_4(s)\ds \revred{ + \varrho_n}\\
& = {z}(t_n) + A_n + (B_n - C_n) + C_n \ls + D_n - E_n \revred{ + \varrho_n}.
\end{align}
Therefore
\begin{align*}
\nom{\x{n}-{z}(t_n)} & \leq \nom{A_n} + \nom{B_n- C_n} + \nom{C_n} \ls + \nom{D_n} + \nom{E_n} \revred{ + \nom{\varrho_n}}.
\end{align*}
Also
\begin{align}
\E\bs{\norm{\x{n}-{z}(t_n)}}^{1/2} & \leq \E\bs{\norm{A_n}}^{1/2} + \E\bs{\norm{E_n}}^{1/2} + \E\bs{\norm{C_n}}^{1/2} + \nonumber \\
& \ \  \E\bs{\norm{D_n}}^{1/2} + \E\bs{\norm{B_n - C_n}}^{1/2} \revred{ + \E\bs{\norm{\varrho_n}}^{1/2}}.\label{eqn: dev}
\end{align}
We shall individually bound the above error terms in the next section under the important assumption of \textit{exponential stability} of the equation of variation (\ref{eqn: lin}): \\

\noindent \textbf{$(\dagger)$} There exists a $\beta > 0$ such that $\forall ~t > s \geq 0$ and $x_0, y_0$,
\[
\nom{\Phi(t,s; x_0, y_0)} \leq \cp e^{-\beta(t-s)}.
\]

This seemingly restrictive  assumption requires some discussion, we argue in particular that some such assumption is \textit{essential} if one is to obtain bounds valid for all time.\\

 To begin, since the idea is to have the parametrized o.d.e.\ (\ref{param}), which is a surrogate for the original iteration, track its unique asymptotically stable equilibrium parametrized by $y$ as the parameter $y \approx y(t)$ changes slowly, it is essential that its rate of approach to the equilibrium, dictated by the spectrum of its linearized drift at this equilibrium, should be much faster than the rate of change of the parameter. This already makes it clear that there will be a requirement of minimum time scale separation for tracking to work at all.\\

 A stronger motivation comes from the fact that the tracking error, given exactly by the Alekseev formula, depends on the linearization of the o.d.e.\ itself around its ideal trajectory $z(\cdot)$, which is a time-varying linear differential equation of the type $\dot{r}(t) = A(t)r(t)$. It is well known in control theory that this can be unstable even if the matrix $A(t)$ is stable for each $t$, see, e.g., Example 8.1, p.\ 131, \cite{Rugh}. Stability is guaranteed to hold only in the special case of $A(t)$ varying slowly with time. The most general result in this direction is that of \cite{solo1994stability}, which we recall below as a sufficient condition  for $(\dagger)$. (There have also been  some extensions thereof to nonlinear systems, see, e.g., \cite{Peut}.)
\newline
\newline
Consider the following time varying linear dynamical system:
\begin{align}
\label{eqn: persolo}
\dot{x}(t) = [A(t) + P(t)] x(t)
\end{align}
and assume the following for this perturbed system:
\begin{enumerate}
\item There exists $\bar{A} > 0$ such that \[ \limsup_{T\uparrow\infty}\frac{1}{T}\int_{t_0}^{t_0+T}\nom{A(s)}\ds \leq \bar{A} \ ~\forall t_0. \]
\item There exists $\gamma \in (0, 1], b > 0$ and $\beta > 0$ sufficiently small in the sense made precise in the theorem below, such that \[ \sum_{t = n_0}^{n_0+n}\nom{A(t_2 + (t-1)T) - A(t_1 + (t-1)T)} \leq Tb + T^{\gamma}(n+1)\beta \ ~\forall n, n_0, \] whenever $|t_2 - t_1| \leq T$.
\item Let $\alpha(t)$ be the real part of the eigenvalue of $A(t)$ whose real part is the largest in absolute value. Then there exists $\bar{\alpha} < 0$ such that, for any $T > 0$,
\[ \limsup_{N\uparrow\infty}\frac{1}{N}\sum_{n = n_0}^{n_0+N}\alpha(s + nT) \leq \bar{\alpha} \ ~\forall s, n_0. \]
\item There exists $\delta > 0$ such that \[ \limsup_{T\uparrow\infty}\int_{t_0}^{t_0+T}\nom{P(s)}\ds \leq \delta, ~ \forall ~ t_0.\]
\end{enumerate}

\begin{thm}[Stability test for deterministic perturbations using eigenvalue based characterization \cite{solo1994stability}]
If the previously mentioned assumptions $(A_1)-(A_4)$ hold, the system $\dot{x}(t) = (A(t) + P(t))x(t)$ is exponentially stable provided we chose $\epsilon, \delta > 0$  small enough so that \[ \bar{\alpha} + \epsilon < 0,\] and
\[
\bar{\alpha} + \epsilon + M_{\epsilon}\delta <0,
\]
with $M_{\epsilon} = 3(\frac{2(\bar{A} + b)}{\epsilon} + 1)^{p-1}/2$, where $\bar{A}, b, \bar{\alpha}$ are as defined in ($A_1$)-($A_4$) and $\beta$ is small enough so that:
\[
\bar{\alpha} + \epsilon + M_{\epsilon} \delta + 2(\ln M_{\epsilon})^{\gamma/(\gamma + 1)}[\beta(M_{\epsilon} + \epsilon/(\bar{A} + b))]^{1/(\gamma + 1)} < 0.
\]
\end{thm}

\bigskip

The correspondence of the foregoing with our framework  is given by $A(\cdot) \leftrightarrow D( \cdot, \cdot)$, $P(\cdot) \leftrightarrow \Xi(\cdot)$.\\

We note here that there are also some sufficient conditions for stability of time-varying linear systems in terms of Liapunov functions, e.g., \cite{zhou2016asymptotic}, \cite{zhou2017stability}, but they appear not so easy to verify.

\section{Error bounds}

Here we obtain the error bounds through a sequence of lemmas.

\subsection{Bound on $D_n$}
\begin{lem}For $D_n$ defined in (\ref{eqn: dndef}),
\begin{align}
\E\bs{\norm{D_n}}^{1/2} \leq \frac{\cp \varepsilon^*}{\beta}
\end{align}
\end{lem}
\begin{proof}
We have
\begin{align}
\nom{D_n} & = \nom{\sum_{k=0}^{n-1}\int_{t_k}^{t_{k+1}}\Phi(t_n,s;\bar{x}(s),y(s))\varepsilon_{k+1}\ds} \nonumber \\
& \leq \varepsilon^* \sum_{k=0}^{n-1}\int_{t_k}^{t_{k+1}}\nom{\Phi(t_n,s;\bar{x}(s),y(s))}\ds \label{eq: dn1}\\
& \leq \cp \varepsilon^* \int_{t_0}^{t_{n}} e^{-\beta (t_n-s)}\ds \label{eq: dn2}\\
& \leq \frac{\cp \varepsilon^*}{\beta}, \nonumber
\end{align}
where (\ref{eq: dn1}) and (\ref{eq: dn2}) follow from (\ref{eq: epbound}) and \textbf{$(\dagger)$} respectively.
Therefore, for all $n\geq 0$, we have
\begin{align}
\E\bs{\norm{D_n}}^{1/2} \leq \frac{\cp \varepsilon^*}{\beta}.
\end{align}
\end{proof}
\subsection{Bound on $E_n$}
\begin{lem}For $E_n$ defined in (\ref{eqn: endef}),
\begin{align}
\E\bs{\norm{E_n}}^{1/2} \leq \frac{K_{\gamma}L_{\lambda}\cp\ep}{\beta}
\end{align}
where $K_{\gamma} := \max_{\|y\| \leq C_{\gamma}}\nom{\gamma(y)}$.
\end{lem}

\begin{proof}
We have,
\begin{align}
\nom{E_n} & = \nom{\sum_{k=0}^{n-1}\int_{t_k}^{t_{k+1}}\Phi(t_n,s;x(s),y(s)) \ls \hqquad \hqquad \times \ep\nabla\lambda(\y{k})\gamma(\y{k})\ds} \nonumber\\
& \leq \ep \sum_{k=0}^{n-1}\int_{t_k}^{t_{k+1}} \nom{\Phi(t_n,s;x(s),y(s))} \ls \hqquad \times \nom{\nabla\lambda(\y{k})} \times \nom{\gamma(\y{k})}\ds,\\
& \leq \ep K_{\gamma}L_{\lambda}\sum_{k=0}^{n-1}\int_{t_k}^{t_{k+1}} \nom{\Phi(t_n,s;x(s),y(s))}\ds.
\end{align}
Using \textbf{$(\dagger)$},
\begin{align}
\nom{E_n} 
& \leq \ep K_{\gamma}L_{\lambda} \cp \int_{t_0}^{t_{n}} e^{-\beta (t_n-s)}\ds\\
& \leq \frac{K_{\gamma}L_{\lambda} \cp \ep }{\beta}.
\end{align}
Hence
\begin{align}
\E\bs{\norm{E_n}}^{1/2} & \leq \frac{K_{\gamma}L_{\lambda}\cp\ep}{\beta}.
\end{align}
\end{proof}
\subsection{Bound on $A_n$}
The next lemma is a variant of Lemma 5.5 of \cite{thoppe2015concentration}.

\begin{lem} For a suitable constant $K_1 > 0$,
\begin{align*}
\int_{t_k}^{t_{k+1}} e^{-\beta(t_n - s)}\nom{\bar{x}(s)-\x{k}}\ds & \qquad \\ \qquad \leq [K_1 &+G_{\tilde{h}}\nom{\x{k}} + \nom{M_{k+1}}]e^{-\beta(t_n-t_{k+1})}a^2.
\end{align*}
\end{lem}

\begin{proof}
Using (\ref{eqn: inter}) and (\ref{SA}), for $s \in [t_k, t_{k+1}]$,
\begin{align}
\nom{\bar{x}(s)-\x{k}} & = \frac{s-t_k}{a}\nom{\x{k+1}-\x{k}} \nonumber \\
& = (s-t_k)[\nom{\tilde{h}(\x{k},\y{k}) + M_{k+1} + \varepsilon_{k+1}}] \nonumber\\
& \leq (s-t_k)[\nom{\tilde{h}(\x{k},\y{k})} + \nom{M_{k+1}} \ls + \nom{\varepsilon_{k+1}}] \nonumber\\
& \leq (s-t_k)[G_{\tilde{h}}(1+\nom{\x{k}} + \nom{\y{k}}) + \nom{M_{k+1}} + \varepsilon^*] \label{eq: htildebound}\\
& \leq (s-t_k)[G_{\tilde{h}}(1 + C_{\gamma}) + G_{\tilde{h}}\nom{\x{k}} + \nom{M_{k+1}} + \varepsilon^*] \nonumber\\
& \leq (s-t_k)[K_1 + G_{\tilde{h}}\nom{\x{k}}+ \nom{M_{k+1}}],
\end{align}
where $K_1 = G_{\tilde{h}}(1 + C_{\gamma})+\varepsilon^*$.
Also,
\begin{align*}
\int_{t_k}^{t_{k+1}} (s-t_k)e^{-\beta(t_n - s)} \ds \leq e^{-\beta(t_n-t_{k+1})}a^2.
\end{align*}
Therefore
\begin{align*}
& \int_{t_k}^{t_{k+1}} e^{-\beta(t_n - s)}\nom{\bar{x}(s)-\x{k}}\ds \ls \ \leq \\
&  \ \ \ \ \ \ \ \ \ \ [K_1 +G_{\tilde{h}}\nom{\x{k}}+ \nom{M_{k+1}}]e^{-\beta(t_n-t_{k+1})}a^2.
\end{align*}
\end{proof}
\begin{lem}For $A_n$ as defined in (\ref{eqn: andef}),
\[
\E\bs{\norm{A_n}}^{1/2} = O(a).
\]
\end{lem}
\begin{proof}
We have
\begin{align}
\nom{A_n} & = \nom{\sum_{k=0}^{n-1}\int_{t_k}^{t_{k+1}}\Phi(t_n,s;x(s),y(s)) \ls \hqquad \hqquad \times \bs{\tilde{h}(\x{k},\y{k})-\tilde{h}(\bar{x}(s),y(s))}\ds} \nonumber \\
& \leq \sum_{k=0}^{n-1}\int_{t_k}^{t_{k+1}}\nom{\Phi(t_n,s;x(s),y(s))} \ls \hqquad \nonumber \\
& \ \ \ \ \ \ \ \  \times \ \nom{\bs{\tilde{h}(\x{k},\y{k}) -\tilde{h}(\bar{x}(s),y(s))}}\ds \nonumber\\
& \leq \sum_{k=0}^{n-1}\int_{t_k}^{t_{k+1}} \cp e^{-\beta(t_n - s)}L_{\tilde{h}}\times \bs{\nom{\x{k} - \bar{x}(s)} \ls \hqquad \hqquad + \nom{\y{k} - y(s)}}\\
& \leq \cp L_{\tilde{h}}\sum_{k=0}^{n-1}\bbs{[K_1 +G_{\tilde{h}}\nom{\x{k}}+ \nom{M_{k+1}}]e^{-\beta(t_n-t_{k+1})}a^2 \statls \hqquad \hqquad + K_{\gamma}a\ep \int_{t_k}^{t_{k+1}}(s - t_k)e^{-\beta(t_n-s)}\ds} \label{eqn: an2}\\
& \leq \cp L_{\tilde{h}}\sum_{k=0}^{n-1}\bbs{[K_1 +G_{\tilde{h}}\nom{\x{k}}+ \nom{M_{k+1}}]e^{-\beta(t_n-t_{k+1})}a^2 \statls \hqquad \hqquad + K_{\gamma}a\ep e^{-\beta(t_n-t_{k+1})}a^2} \label{eqn: an3} \\
& = \cp L_{\tilde{h}}\sum_{k=0}^{n-1} [K_1 + G_{\tilde{h}}\nom{\x{k}}+ K_{\gamma}a\ep + \nom{M_{k+1}}]\ls \hqquad \hqquad \hqquad \hqquad \times e^{-\beta(t_n-t_{k+1})}a^2 \nonumber\\
& \leq a \cp L_{\tilde{h}}\bbc{(K_1 + K_{\gamma}a\ep) \mu \ls \hqquad \ + \nonumber \\
& \ \ \ \ \ \ \sum_{k=0}^{n-1}\bc{G_{\tilde{h}}\nom{\x{k}}+\nom{M_{k+1}}}a e^{-\beta(t_n-t_{k+1})}}, \nonumber
\end{align}
where $\mu := \frac{1}{\beta}$. Equation (\ref{eqn: an2}) follows from Lemma 3. Denote the terms $\cp L_{\tilde{h}}$, $K_1 + K_{\gamma}a\ep$, $\sum_{k=0}^{n-1} \nom{M_{k+1}}a e^{-\beta(t_n-t_{k+1})}$ and $\sum_{k=0}^{n-1}a\nom{\x{k}} e^{-\beta(t_n-t_{k+1})}$ by $K_2$, $K_3$, $F_n$ and $\tilde{F}_n$. Note that $K_3$ is $O(1)$. Then
\begin{align*}
\nom{A_n} & \leq aK_2\bc{K_3\mu + F_n + G_{\tilde{h}}\tilde{F}_n},\\
\E[\norm{A_n}]^{1/2} & \leq aK_2\bbc{K_3\mu + \E[F_n^2]^{1/2} + G_{\tilde{h}}\E[\tilde{F}_n^2]^{1/2}}.
\end{align*}
Now,
\begin{align*}
\E[\tilde{F}_n^2]^{1/2} & = \sum_{k=0}^{n-1}\bbc{a\E[\norm{\x{k}}]^{1/2}e^{-\beta(t_n-t_{k+1})}}\\
& \leq C^*\mu\numberthis \label{eq:fntildeBound}
\intertext{and}
\fls{\E[F_n^2]^{1/2}}
& = \sum_{k=0}^{n-1}\E[\norm{M_{k+1}}]^{1/2}ae^{-\beta(n-k)a}\\
& \leq \sum_{i=0}^{n-1}\frac{\sqrt{2\cm}}{\delta}ae^{-\beta(n-k)a} \\
&\leq K_4 \mu \numberthis \label{eq:fnsqrBound}
\end{align*}
where $K_4 = \frac{\sqrt{2\cm}}{\delta}$.
Therefore
\begin{align*}
\E[\norm{A_n}]^{1/2} & \leq aK_2\bbc{K_3\mu + K_4\mu + G_{\tilde{h}}C^*a\mu}.
\end{align*}
Hence $\E[\norm{A_n}]^{1/2} = O(a).$
\end{proof}
\subsection{Bound on $B_n - C_n$}

\begin{lem}[]For $B_n$ and $C_n$ defined in (\ref{eqn: bndef}) and (\ref{eqn: cndef})
$$\E[\norm{B_n-C_n}]^{1/2} = O(a).$$
\end{lem}

\begin{proof}
From (\ref{eqn: bndef}) and (\ref{eqn: cndef}) we have
\begin{align*}
\fls{B_n - C_n}
 & =\sum_{k=0}^{n-1}\int_{t_k}^{t_{k+1}}\bs{\Phi(t_n,s;\bar{x}(s),y(s)) \ls \hqquad -\Phi(t_n,s;\x{k},\y{k})}M_{k+1}\ds
 \end{align*}
 Therefore
 \begin{align*}
\fls{\nom{B_n - C_n}}
 & = \nom{\sum_{k=0}^{n-1}\int_{t_k}^{t_{k+1}}\bs{\Phi(t_n,s;\bar{x}(s),y(s)) \ls \hqquad -\Phi(t_n,s;\x{k},\y{k})}M_{k+1}\ds}\\
 & \leq \sum_{k=0}^{n-1}\int_{t_k}^{t_{k+1}}\nom{\Phi(t_n,s;\bar{x}(s),y(s)) \ls -\Phi(t_n,s;\x{k},\y{k})} \\
 & \ \ \ \ \ \ \ \ \ \ \ \ \ \ \ \ \ \times \ \nom{M_{k+1}}\ds.
\end{align*}
From (\ref{eqn: lin}), we know that $\Phi(t,s;\bar{x}(s),y(s))$ and $\Phi(t,s;\x{k},\y{k})$ are fundamental matrices for the linear systems given by: for $t \geq s$,
\begin{align}
\fls{\dot{\chi}(t,s;\bar{x}(s),y(s)) =} & D(z(t,s;\bar{x}(s),y(s)), y(t,s;y(s)))\chi(t,s;\bar{x}(s),y(s)), \label{eq: eqvar1}
\intertext{and}
\fls{\dot{\widetilde{\chi}}(t,s;\x{k},\y{k}) =} & D(z(t,s;\x{k},\y{k}), y(t,s;\y{k}))\widetilde{\chi}(t,s;\x{k},\y{k}). \label{eq: eqvar2}
\end{align}
So $\Phi(t,s;\bar{x}(s),y(s))$ and $\Phi(t,s;\x{k},\y{k})$ satisfy the following matrix valued differential equations
\begin{align}
\fls{\dot{\Phi}(t,s;\bar{x}(s),y(s)) =} & D(z(t,s;\bar{x}(s),y(s)), y(t,s;y(s)))\Phi(t,s;\bar{x}(s),y(s)), \label{eq: matdiff1}
\intertext{and}
\fls{\dot{\Phi}(t,s;\x{k},\y{k}) =} & D(z(t,s;\x{k},\y{k}), y(t,s;\y{k}))\Phi(t,s;\x{k},\y{k}). \label{eq: matdiff2}
\end{align}
For each column indexed by $j$, the differential equations (\ref{eq: matdiff1}) and (\ref{eq: matdiff2}) can be equivalently written as
\begin{align}
\fls{\dot{\Phi}_j(t,s;\bar{x}(s),y(s))} & = D(z(t,s;\bar{x}(s),y(s)), y(t,s;y(s)))\Phi_j(t,s;\bar{x}(s),y(s)), \label{eq: matdiff3}
\intertext{and}
\fls{\dot{\Phi}_j(t,s;\x{k},\y{k})} & = D(z(t,s;\bar{x}(s),y(s)), y(t,s;y(s)))\Phi_j(t,s;\x{k},\y{k}) \ + \ \statls \bs{D(z(t,s;\x{k},\y{k}), y(t,s;\y{k})) \statls - D(z(t,s;\bar{x}(s),y(s)), y(t,s;y(s)))} \statls \times \Phi_j(t,s;\x{k},\y{k}). \label{eq: matdiff4}
\end{align}
Treating (\ref{eq: matdiff4}) as a perturbation of (\ref{eq: matdiff3}) and applying Alexseev's formula (\ref{eq:alekseev})\footnote{in fact, the classical variation of constants formula for linear systems which it generalizes} to each column of $\Phi(\bullet,\bullet;\bullet,\bullet)$, we have
\begin{align}
& \Phi_j(t_n,s;\x{k},\y{k}) - \Phi_j(t_n,s;\bar{x}(s),y(s))
\statls = \int_{s}^{t_n} \Phi(t_n,t;\bar{x}(t),y(t))\statls \times \bs{D(z(t,s;\x{k},\y{k}),y(t,s;\y{k})) \ls -D(z(t,s;\bar{x}(s),y(s)),y(t,s;y(s)))}\statls \times \Phi_j(t,s;\x{k},\y{k})\dt \label{eq: matdiff5}
\intertext{Combining the equations (\ref{eq: matdiff5}) for all columns, we get}
& \Phi(t_n,s;\x{k},\y{k}) - \Phi(t_n,s;\bar{x}(s),y(s))
\statls = \int_{s}^{t_n} \Phi(t_n,t;\bar{x}(s),y(s))\statls \times \bs{D(z(t,s;\x{k},\y{k}),y(t,s;\y{k})) \ls -D(z(t,s;\bar{x}(s),y(s)),y(t,s;y(s)))}\statls \times \Phi(t,s;\x{k},\y{k})\dt \nonumber
\end{align}
Therefore
\begin{align}
\fls{\nom{B_n - C_n}}
& \leq \sum_{k=0}^{n-1}\int_{t_k}^{t_{k+1}}\int_{s}^{t_n} \nom{\Phi(t_n,t;\bar{x}(s),y(s))}\statls \times \nom{\bs{D(z(t,s;\x{k},\y{k}),y(t,s;\y{k})) \statls -D(z(t,s;\bar{x}(s),y(s)),y(t,s;y(s)))}}\statls \times \nom{\Phi(t,s;\x{k},\y{k})} \times \nom{M_{k+1}}\dt ~\ds \nonumber \\
& \leq\sum_{k=0}^{n-1}\int_{t_k}^{t_{k+1}}\int_{s}^{t_n}\cp^2L_D \times e^{-\beta(t_n-s)} \times e^{-\beta(t-s)} \times \statls \bbs{\nom{z(t,s;\bar{x}(s),y(s)) - z(t,s;\x{k},\y{k})} \statls + \nom{y(s)-\y{k}}} \times \nom{M_{k+1}}\dt~\ds \label{eq: bncn1}\\
& \leq\sum_{k=0}^{n-1}\int_{t_k}^{t_{k+1}}\int_{s}^{t_n}\cp^2L_D \times e^{-\beta(t_n-s)} \times e^{-\beta(t-s)} \times \statls \bbs{\bs{\nom{\bar{x}(s)-\x{k}} + \nom{y(s)-\y{k}}}\cp e^{-\beta(t-s)} \statls + \nom{y(s)-\y{k}} } \times \nom{M_{k+1}}\dt~\ds, \label{eq: bncn2}
\end{align}
where 
(\ref{eq: bncn1}) follows from \textbf{$(\dagger)$} and the Lipschitz property of $D(\cdot,\cdot)$
while (\ref{eq: bncn2}) follows from (\ref{eq:genalekseev}) and $(\dagger)$. We split the analysis into two terms as follows:
\begin{align*}
G_n & := \sum_{k=0}^{n-1}\int_{t_k}^{t_{k+1}}\int_{s}^{t_n}\cp^3L_D e^{-\beta(t_n-s)}\ls \nom{\bar{x}(s)-\x{k}} \times e^{-2\beta(t-s)} \times \nom{M_{k+1}}\dt~\ds\\
& = \sum_{k=0}^{n-1}\int_{t_k}^{t_{k+1}}\cp^3L_De^{-\beta(t_n-s)}\nom{\bar{x}(s)-\x{k}}\ls\times\frac{1-e^{-2\beta(t_n-s)}}{2\beta} \times \nom{M_{k+1}}\ds\\
& = K_5\sum_{k=0}^{n-1}\int_{t_k}^{t_{k+1}}e^{-\beta(t_n-s)} \times \nom{\bar{x}(s)-\x{k}} \ls\times \bc{1-e^{-2\beta(t_n-s)}}\times \nom{M_{k+1}}\ds \\
& \leq K_5\sum_{k=0}^{n-1}\int_{t_k}^{t_{k+1}}e^{-\beta(t_n-s)}\nom{\bar{x}(s)-\x{k}}\times\nom{M_{k+1}}\ds\\
& \leq K_5\sum_{k=0}^{n-1}a^2 e^{-\beta(t_n-t_{k+1})}\bs{K_1\nom{M_{k+1}} + \norm{M_{k+1}}+G_{\tilde{h}}\nom{\x{k}} \ \nom{M_{k+1}}} \numberthis \label{eq: gnXbound}
\intertext{where $K_5$ denotes $\cp^3L_D/2\beta$ and (\ref{eq: gnXbound}) follows from Lemma 3,}
H_n & := \sum_{k=0}^{n-1}\int_{t_k}^{t_{k+1}}\int_{s}^{t_n}\cp^3L_D \times e^{-\beta(t_n-s)} \nom{y(s)-\y{k}}\ls e^{-2\beta(t-s)}\nom{M_{k+1}}\dt\ds\\
& + \  \sum_{k=0}^{n-1}\int_{t_k}^{t_{k+1}}\int_{s}^{t_n}\cp^2L_D \times e^{-\beta(t_n-s)} \nom{y(s)-\y{k}}\ls e^{-\beta(t-s)}\nom{M_{k+1}}\dt\ds\\
& = \sum_{k=0}^{n-1}\int_{t_k}^{t_{k+1}}\cp^3L_D e^{-\beta(t_n-s)}\nom{y(s)-\y{k}} \ls \bbs{\frac{1}{2\beta} \bc{1-e^{-2\beta(t_n-s)}}}\nom{M_{k+1}}\ds\\
& + \ \sum_{k=0}^{n-1}\int_{t_k}^{t_{k+1}}\cp^2L_D e^{-\beta(t_n-s)}\nom{y(s)-\y{k}} \ls \bbs{\frac{1}{\beta} \bc{1-e^{-\beta(t_n-s)}}}\nom{M_{k+1}}\ds\\
& \leq \sum_{k=0}^{n-1}\int_{t_k}^{t_{k+1}}\frac{\cp^2(\frac{1}{2} + \cp)L_D}{\beta}e^{-\beta(t_n-s)}K_{\gamma}a\ep (s-t_k)\nom{M_{k+1}}\ds \\
& \leq K_6\sum_{k=0}^{n-1}e^{-\beta(t_n-t_{k+1})}a^3\nom{M_{k+1}} \numberthis \label{eq: hnKbound}
\end{align*}
where (\ref{eq: hnKbound}) follows from (\ref{slow}) and $K_6 :=(\frac{1}{2}+ \cp)\cp^2K_{\gamma}L_D\ep/\beta$. Further define $G_{1,n}$, $G_{2,n}$ and $G_{3,n}$ as follows
\begin{align*}
G_{1,n} & = \sum_{k=0}^{n-1}a e^{-\beta(t_n-t_{k+1})}\nom{M_{k+1}},\\
G_{2,n} & = \sum_{k=0}^{n-1}a e^{-\beta(t_n-t_{k+1})}\norm{M_{k+1}},\\
G_{3,n} & = \sum_{k=0}^{n-1}a e^{-\beta(t_n-t_{k+1})}\nom{\x{k}} \ \nom{M_{k+1}}.
\end{align*}
Then
\begin{align*}
\nom{B_n - C_n} & \leq G_n + H_n \\
& \leq K_5a(K_1G_{n,1} + G_{n,2}+G_{\tilde{h}}G_{n,3}) + \ls \bc{K_6a^2G_{n,1}}.\\
\intertext{Therefore}
\E[\norm{B_n - C_n}]^{1/2} & \leq K_5a\bc{K_1\E[G_{n,1}^2]^{1/2} + \E[G_{n,2}^2]^{1/2} +G_{\tilde{h}}\E[G_{n,3}^2]^{1/2}}+ \statls a^2K_6\E[G_{n,1}^2]^{1/2}
\end{align*}
We now bound each of the terms in the previous expression. Using a calculation similar to the one used for (\ref{eq:fnsqrBound}), we have
\begin{align*}
& \E[G_{n,1}^2]^{1/2} \leq K_4\mu, \numberthis \label{eq:gn1hn1Bound}\\
& \E[G_{n,2}^2]^{1/2}  = \sum_{k=0}^{n-1} E[||M_{k+1}||^4]^{1/2} a e^{-\beta(n-k)a} \\
& \leq \sum_{k=0}^{n-1} \dfrac{\sqrt{24\cm}}{\delta^2} a e^{-\beta(n-k)a}
\\
& \leq K_7 \bbc{\sum_{k=0}^{n-1} a e^{-\beta(n-k)a}} \\
& \leq K_7\mu, \numberthis \label{eq:gn2Bound}  \\
\intertext{where $K_7 = \sqrt{24\cm}/ \delta^2$,} \\
\E[G_{n,3}^2]^{1/2} & = \E\bbs{\sum_{k=0}^{n-1}a e^{-\beta(t_n-t_{k+1})}E\left[\norm{\x{k}}\norm{M_{k+1}}\right]^{1/2}}\\
& \leq \sum_{k=0}^{n-1}a e^{-\beta(t_n-t_{k+1})}\bbc{\E\bs{\nom{\x{k}}^4}\E\bs{\nom{M_{k+1}}^4}}^{1/4}\\
& \leq \sum_{k=0}^{n-1}a e^{-\beta(t_n-t_{k+1})}C^*\bbc{\E\bs{\nom{M_{k+1}}^4}}^{1/2}\\
& \leq \sum_{k=0}^{n-1}ae^{-\beta(t_n-t_{k+1})}C^*\left(\frac{(24\cm)^{1/4}}{\delta}\right)\\
& \leq K_8\mu, \numberthis \label{eq:gn3Bound}
\end{align*}
where $K_8 = \dfrac{C^*(24\cm)^{1/4}}{\delta}$.

Using (\ref{eq:gn1hn1Bound}), (\ref{eq:gn2Bound}) and (\ref{eq:gn3Bound}), we have
\begin{align*}
\E[\norm{B_n - C_n}]^{\frac{1}{2}} & \leq K_5a\bc{K_1K_4\mu + K_7\mu + G_{\tilde{h}}K_8\mu }+ a^2K_6K_8\mu\\
& = O(a) \numberthis \label{eq:bncnBound}.
\end{align*}
\end{proof}

\subsection{Bound on $C_n$}

\begin{lem}[]For $C_n$ defined in (\ref{eqn: cndef})
$$\E[\norm{C_n}]^{1/2} = O(\max \{a^{1.5}d^{3.25}, a^{0.5}d^{2.5}\}).$$
\end{lem}

\begin{proof}
It is easy to verify that $C_n$ satisfies the condition for the martingale concentration inequality provided in Theorem \ref{thm:concineq} in Appendix, with
\begin{align*}
\alpha_{k,n} & = \int_{t_k}^{t_{k+1}} \Phi(t_n,s,\x{k},\y{k})\ds,\\
\gamma_1 & = \frac{\cp}{\beta}, \ \gamma_2 = 1, \
\beta_n  = a,
\end{align*}
for $k, n \geq 0$. Thus,
\begin{align*}
E[||C_n||^2] & = \int_{0}^{\infty} P(\norm{C_n} \geq s)\ds\\
& = \int_{0}^{\infty} P(\nom{C_n} \geq \sqrt{s} )\ds
\end{align*}
Using the martingale concentration inequality provided in \ref{thm:concineq} in Appendix, we have
\begin{align}
E[||C_n||^2] & = \int_{0}^{K_{9}} 2d^2\text{exp}\Big(\dfrac{-cs}{d^3a}\Big)\ds \nonumber \\
& +\int_{K_{9}}^{\infty} 2d^2\text{exp}\Big(\dfrac{-c\sqrt{s}}{d^{3/2}a}\Big)\ds,   \label{eq:cnIntermediate}
\end{align}
where $K_{9} = \dfrac{\cm\gamma_1d^{1.5}}{\delta}$. Analysing the  terms separately, we have
\begin{align*}
\int_{0}^{K_{9}} 2d^2\text{exp}\Big(\dfrac{-cs}{d^3a}\Big)\ds &= \dfrac{2d^5a}{c}\Big(1-\text{exp}\Big(\dfrac{-cK_{9}}{d^3a}\Big)\Big) \\
&\leq \Big(\dfrac{2d^5}{c}\Big)a,  \numberthis \label{eq:sn1}
\end{align*}
and
\begin{align*}
\fls{\int_{K_{9}}^{\infty} 2d^2\text{exp}\Big(\dfrac{-c\sqrt{s}}{d^{3/2}a}\Big)\ds}
&= \dfrac{4d^5a}{c^2}\text{exp}\Big(\frac{-c\sqrt{K_{9}}}{d^{3/2}a}\Big)\Big(a + \dfrac{c\sqrt{K_{9}}}{d^{3/2}}\Big) \\
&\leq \dfrac{4d^5a}{c^2}\Big(\dfrac{ad^{3/2}}{c\sqrt{K_{9}}}\Big)\Big(a + \dfrac{c\sqrt{K_{9}}}{d^{3/2}}\Big) \numberthis \label{eq:expBound} \\
&= O(\max \{a^3d^{6.5}, a^2d^5\}), \numberthis \label{eq:sn2} 
\end{align*}
where (\ref{eq:expBound}) follows from the fact that $e^{-1/a}\leq a$ for $a > 0$. From (\ref{eq:sn1}) and (\ref{eq:sn2}), we have
\begin{align*}
\E[\norm{C_n}] & \leq \bbc{\frac{2d^5}{c}}a + O(\max \{a^3d^{6.5}, a^2d^5\})\\
& = O(\max \{a^3d^{6.5}, ad^5\})\\
\therefore  \ \ \E[\norm{C_n}]^{1/2} & = O(\max \{a^{1.5}d^{3.25}, a^{0.5}d^{2.5}\}).
\end{align*}
\end{proof}


\section{Main result}

Combining the foregoing bounds leads to our main result stated as follows.\\

\begin{thm}
\label{thm: mainthm}
The mean square deviation of tracked iterates from a non-stationary trajectory satisfies:
\begin{align}
\label{eqn: main}
\fls{\E\bs{\norm{x_{n}- \lambda(y(n))}}^{1/2}}
& \leq \frac{\cp \varepsilon^*}{\beta} + \frac{K_{\gamma}L_{\lambda}\cp \ep}{\beta} \statls + O(\max \{a^{1.5}d^{3.25}, a^{0.5}d^{2.5}\}) \statls \revred{ + \cp e^{-\beta(t_n-t_0)}\nom{x_0 - \lambda(y(0))}}
\end{align}
\end{thm}

\begin{proof}
Using (\ref{eqn: dev}), ($\dagger$) and lemmas 1-6, we get
\begin{align*}
\E\bs{\norm{\x{n}-{z}(t_n)}}^{1/2} & \leq \frac{\cp \varepsilon^*}{\beta} + \frac{K_{\gamma}L_{\lambda}\cp \ep}{\beta} + O(a) \statls + O(\max \{a^{1.5}d^{3.25}, a^{0.5}d^{2.5}\}) \statls  \revred{ + \cp e^{-\beta(t_n-t_0)}\nom{\bar{x}(t_0) - z(t_0)}}\\
& = \frac{\cp \varepsilon^*}{\beta} + \frac{K_{\gamma}L_{\lambda}\cp \ep}{\beta} \statls + O(\max \{a^{1.5}d^{3.25}, a^{0.5}d^{2.5}\}) \statls \revred{ + \cp e^{-\beta(t_n-t_0)}\nom{\bar{x}(t_0) - z(t_0)}}.
\end{align*}
The claim follows.
\end{proof}

\bigskip

\textbf{Remark}: 1. The $O(\cdot)$ notation is used above to isolate the dependence on the stepsize $a$. The exact constants involved are available in the relevant lemmas, but are suppressed in order to improve clarity.\\

2. The linear complexity of the error bound in $\varepsilon^*$ and $\ep$ is natural to expect, these being contributions from bounded additive error component $\varepsilon_n$ and rate of variation of the tracking signal, respectively. The $O(\cdot)$ term is due to the  martingale noise  and discretization. \revred{ The last term accounts for the effect of initial condition.}\\

3. By setting $\ep = 0$ in (\ref{eqn: main}), we can recover as a special case a bound valid for  all time  for a stationary target. Then $y(\cdot) \equiv y^*$, a constant, and $z(\cdot) \equiv x^* = \lambda(y^*),$ also a  constant, viz., an equilibrium  for the system $\dot{x}(t) = h(x(t), y^*)$.

\section{Conclusion and Future Work}
We analyzed a constant step-size stochastic approximation algorithm for tracking a slowly varying dynamical system and obtained a \textit{non-asymptotic} bound \textit{valid for  all time}, with dependence on step-size and dimension explicitly given. The latter in particular provides  insight into step-size selection in high dimensional regime.\\

A natural extension would be to the problem of tracking a stochastic dynamics. Indeed, a suitable extension of Alekseev's formula is available for this purpose \cite{Zerihun}, which is much more complex.


\bigskip

\bigskip

\newpage

\noindent \textbf{\large Appendix: A martingale concentration inequality }\\

\bigskip

We state here the martingale concentration inequality we have used, from \cite{thoppe2015concentration}, which in turn is a slight adaptation of the results of \cite{LiuWatbled}.\\

\begin{thm} \label{thm:concineq}
Let $S_n = \sum_{k=1}^{n}\alpha_{k,n}X_k$, where ${X_k}$ is a $\mR^d$ valued ${\mathcal{F}_k}$ - adapted martingale difference sequence and ${\alpha_{k,n}}$ is a sequence of bounded pre-visible real valued $d \times d$ random matrices, i.e., $\alpha_{k,n} \in \mathcal{F}_{k-1}$ and there exists finite number, say $A_{k,n}$, such that $\nom{\alpha_{k,n}} \leq A_{k,n}$. Suppose that for some $\delta, C >0$
\[
\E\bs{e^{\delta\nom{X_k}}\bigg|\mathcal{F}_{k-1}} \leq C, \ k \geq 1.
\]
Further assume that there exist constants $\gamma_1, \gamma_2 >0$, independent of $n$, so that $\sum_{k=1}^{n}A_{k,n} \leq \gamma_1$ and $\max_{1\leq k \leq n}A_{k,n} \leq \gamma_2\beta_n$, where ${\beta_n}$ is some positive sequence. Then for $\eta > 0$, there exists some constant $c > 0$ depending on $\delta, C, \gamma_1, \gamma_2$ such that \[
P(\nom{S_n} > \eta)
\leq
\begin{cases}
	2d^2e^{-\frac{c\eta^2}{d^3\beta_n}} & if ~ \eta \in \bigg(0,\frac{C\gamma_1d^{1.5}}{\delta}\bigg],\\
	 2d^2e^{-\frac{c\eta}{d^{1.5}\beta_n}} & otherwise.
\end{cases}
\]
\end{thm}


\end{document}